\documentclass[11pt,a4paper]{article}

\usepackage[utf8]{inputenc}
\usepackage[T1]{fontenc}
\usepackage{lmodern}
\usepackage{amsmath,amssymb,amsfonts,amsthm}
\usepackage{xcolor}
\usepackage{booktabs}
\usepackage{microtype}
\usepackage{hyperref}
\usepackage[margin=1in]{geometry}
\usepackage{tikz}
\usetikzlibrary{positioning,arrows.meta}

\newtheorem{theorem}{Theorem}
\newtheorem{lemma}[theorem]{Lemma}
\newtheorem{corollary}[theorem]{Corollary}
\newtheorem{proposition}[theorem]{Proposition}
\theoremstyle{definition}
\newtheorem{definition}{Definition}
\theoremstyle{remark}
\newtheorem{remark}{Remark}

\newcommand{\FC}{\textsf{FC}}
\newcommand{\ID}{\textsf{ID}}
\newcommand{\pillar}[1]{\textbf{P#1}}
\newcommand{\op}[1]{\textnormal{\textsc{#1}}}
\newcommand{\rel}[1]{\textsf{#1}}

\title{\textbf{First-Class Relationships\\ as a Representation Threshold:\\
Identity, Time, Confidence, and Causation}}

\author{Matthew Alford\\
Independent researcher\\[2pt]
\small ORCID: \href{https://orcid.org/0009-0007-3856-8362}{0009-0007-3856-8362}}

\date{November 2026}

\begin{document}
\maketitle

\begin{abstract}
Operational records---who may access what, what changed, what caused an incident---consist of relationships that involve several parties, hold for a period, are believed with some confidence, and cause one another. We ask what a representation must provide to store such records directly. We state four storage requirements: n-ary relationships with attributes, temporal validity, confidence with an evidential basis, and causal links between relationships. Under our definitions, each requires relationships to have identity: every relationship is stored, fetched and queried as an object with its own key. Confidence with an evidential basis and causal links also require composition: any relationship can take part in another. With n-ary relationships, either one therefore requires first-class relationships, which are identified, carry attributes, and compose at every arity. Conversely, first-class relationships suffice for all four. They are thus a threshold, and it belongs to the relationship construct, not to the data model around it: where a model's own construct falls short, every representation meeting n-ary relationships together with evidence-backed confidence or causal links builds a first-class construct from other parts, as our SQL implementation does. We place SQL, property graphs, RDF~1.2 and TypeDB relative to the threshold, and present the Attributed Temporal Causal Hypergraph (ATCH), which meets all four requirements with query operators for each.
\end{abstract}

\section{Introduction}
\label{sec:intro}

An IT services company manages computers and accounts for many small businesses, and its technicians usually hold administrator access to their clients' systems. One technician, Dana, has held administrator access to a client's email since 3 November 2025. She leaves the company on 10 March 2026, and her access should end that day. But on 12 March her account signs in to the client's email and approves a mail app run by an outsider, giving it permission to read and change mail. On 13 March the app adds a rule that forwards every invoice from the finance mailbox to the outsider. On 15 March the client pays a forged invoice into the outsider's bank account.

The investigators ask three questions. Could Dana's account manage the client's email on 12 March? What chain of events links her access to the payment, and how sure is each step? And why does the company's offboarding record say her access ended on 10 March, while the client's sign-in record shows it in use on 12 March?

Each question is about relationships rather than things. Dana's access joins her account, an administrator group and the client's email, and it held for a period whose end is in dispute. The claim that the forwarding rule enabled the payment is a belief, not a fact. And the causal links---her lingering access let her account approve the app, and the approval let the app add the rule---connect relationships to relationships.

Relational tables, property graphs and triple stores can all store this record, but each does so by building something: a table whose rows stand for relationships of any type, a node per relationship, or a reified statement. The device is old: ER aggregation~\cite{smith1977database}, RM/T surrogates~\cite{codd1979extending}, object identity~\cite{khoshafian1986object,abiteboul1989object} and RDF reification~\cite{hartig2017foundations} all provide it in some form. Following them, we call what is being built \emph{first-class relationships}, and make precise which storage requirements force it. Under the definitions of Section~\ref{sec:model}, a representation can store n-ary relationships together with either evidence-backed confidence or causal links between relationships exactly when it provides first-class relationships, natively or by construction. The contribution is the requirement-by-requirement map of what each storage need forces, not the construct itself.\footnote{This is version~2 of arXiv:2603.13603. Version~1 claimed an equivalence theorem for ``structurally complete'' systems; that claim is withdrawn, and the result here is narrower.}

\paragraph{Contributions.}
\begin{enumerate}
    \item \textbf{A threshold result} (Section~\ref{sec:theorem}). Each of four storage requirements needs relationship identity (n-ary relationships need it when equal ones must be kept distinct); confidence with an evidential basis and causal links need composition; with n-ary relationships, either needs first-class relationships; and first-class relationships suffice for all four. Parts of it follow closely from the definitions; we say which in Remark~\ref{rem:scope}.
    \item \textbf{A comparison} (Section~\ref{sec:comparison}) that separates native constructs from built ones and places SQL, property graphs, RDF~1.2, TypeDB and our model consistently with the theorem.
\end{enumerate}
We also observe (Section~\ref{sec:queries}) that the investigators' questions fall in known polynomial-time query classes. The claims concern \emph{representation}: what must be stored, not how causes or confidences are inferred.

\section{Running Example}
\label{sec:example}

Table~\ref{tab:example} lists the four relationships at the core of the incident and the causal links between them. Dates are in 2026 unless stated. A confidence $\kappa$ on a relationship says how sure we are that it held as recorded; a confidence on a causal link says how sure we are of the causal claim.

\begin{table}[ht]
\centering
\small
\begin{tabular}{@{}l p{8.6cm} l l@{}}
\toprule
& \textbf{Participants (role) and attributes} & \textbf{Valid time} & $\kappa$ \\
\midrule
\rel{access} & Dana's account (member), admin group (group), client email (resource); \emph{level}: admin & 3 Nov 2025 -- ? & -- \\
\rel{approval} & Dana's account (approver), outsider's mail app (app), client email (resource); \emph{permission}: read/change mail & from 12 Mar & 1.0 \\
\rel{rule} & mail app (creator), finance mailbox (from), outsider's address (to); \emph{action}: forward invoices & from 13 Mar & 1.0 \\
\rel{payment} & client (payer), vendor (named on invoice), outsider's bank account (paid to) & 15 Mar & 1.0 \\
\midrule
\multicolumn{4}{@{}l}{Causal links, each of kind \emph{enabled}: \rel{access} $\xrightarrow{0.9}$ \rel{approval} $\xrightarrow{0.95}$ \rel{rule} $\xrightarrow{0.7}$ \rel{payment}} \\
\bottomrule
\end{tabular}
\caption{The core of the incident. The payment certainly happened ($\kappa = 1$); what is uncertain is whether the forwarding rule enabled it (0.7). The end of \rel{access}, and so its confidence, is open: two records disagree.}
\label{tab:example}
\end{table}

\paragraph{Two records about the access.} Two records disagree about when \rel{access} ended. The \emph{offboarding record}, from the company's ticketing system, says it ended on 10 March; we give it confidence 0.8. The \emph{sign-in record}, from the client's identity system, shows Dana's account signing in on 12 March; we give it 0.97. Each record is an \emph{assessment} of \rel{access}: a relationship among \rel{access}, a source (the system) and a method (an offboarding ticket, a sign-in log entry), with a value and its own confidence. An assessment may concern any attribute of the relationship, its valid time included, and each record's $\kappa$ is the confidence in that record. We name them \rel{offboarding} and \rel{signin}. Figure~\ref{fig:example} shows the whole incident.

\begin{figure}[ht]
\centering
\begin{tikzpicture}[
  rel/.style={rectangle, draw=blue!50!black, rounded corners=4pt, fill=blue!7,
              text width=2.3cm, align=center, inner sep=4pt, font=\footnotesize},
  asm/.style={rectangle, draw=gray!70, dashed, rounded corners=4pt, fill=white,
              text width=3.0cm, align=center, inner sep=3pt, font=\footnotesize},
  cause/.style={-{Stealth[length=2mm]}, thick, blue!60!black},
  lbl/.style={font=\scriptsize, blue!60!black, align=center},
  assess/.style={-{Stealth[length=1.6mm]}, dashed, gray!70!black}
]
\node[rel] (m) at (0,0)    {\textbf{\textsf{access}}\\ Dana's account,\\ admin group,\\ client email};
\node[rel] (c) at (3.9,0)  {\textbf{\textsf{approval}}\\ Dana's account,\\ mail app, client email};
\node[rel] (r) at (7.8,0)  {\textbf{\textsf{rule}}\\ mail app, mailbox,\\ outsider};
\node[rel] (f) at (11.7,0) {\textbf{\textsf{payment}}\\ client, vendor,\\ outsider's bank};
\draw[cause] (m) -- node[lbl, above] {enabled\\ 0.9} (c);
\draw[cause] (c) -- node[lbl, above] {enabled\\ 0.95} (r);
\draw[cause] (r) -- node[lbl, above] {enabled\\ 0.7} (f);
\node[asm] (a1) at (-0.95,-2.2) {\textbf{\textsf{offboarding}}\\ ended 10 Mar\\ $\kappa$ 0.8};
\node[asm] (a2) at (2.45,-2.2)  {\textbf{\textsf{signin}}\\ in use 12 Mar\\ $\kappa$ 0.97};
\draw[assess] (a1) -- node[font=\scriptsize, left=2pt] {assesses} (m);
\draw[assess] (a2) -- node[font=\scriptsize, right=2pt] {assesses} (m);
\end{tikzpicture}
\caption{The incident as first-class relationships. Each box is a relationship with its own participants. Solid arrows are causal links between relationships, with their confidence; dashed boxes are records about \rel{access}, themselves relationships with \rel{access} as a participant.}
\label{fig:example}
\end{figure}

\paragraph{The investigators' questions.} (Q1) Could Dana's account manage the client's email on 12 March, according to each record? Joining each assessment to \rel{access} answers it: \rel{signin} says yes, \rel{offboarding} says no. (Q2) What chain of causal links leads from \rel{access} to \rel{payment}, and how confident is it? In this incident there is one chain, $\textsf{access} \to \textsf{approval} \to \textsf{rule} \to \textsf{payment}$, with confidence $0.9 \times 0.95 \times 0.7 \approx 0.60$ under the product rule of Section~\ref{sec:queries}; in general there may be many, and the question asks for the most confident. (Q3) Why is the end of \rel{access} disputed? Its two assessments answer: different sources, methods and dates.

\paragraph{A construct built in SQL.} With one table per relationship type, a causal link must refer to rows of different tables:

\begin{center}
\small
\begin{minipage}{0.86\textwidth}
\begin{verbatim}
CREATE TABLE causal_link (
  cause_table  TEXT,   cause_id  BIGINT,  -- access? approval? rule?
  effect_table TEXT,   effect_id BIGINT,  -- no foreign key possible
  kind         TEXT,   confidence FLOAT);
\end{verbatim}
\end{minipage}
\end{center}

\noindent The \texttt{(table, id)} pair is a reference to a relationship of any type, assembled by hand and unchecked by the database. The assessments need the same device. The cleaner alternative is a single table \texttt{relationship(id, \ldots)} that every relationship type references, with a membership table whose participants may be relationships. That is what our own implementation does (Section~\ref{sec:queries}), and Corollary~\ref{cor:built} shows that some such built reference is unavoidable.

\section{Model}
\label{sec:model}

\subsection{Relationship constructs}

A representation stores relationships as objects of some \emph{relationship construct}: a collection of objects that stand for relationships, each with participants in roles, together with a \emph{reference domain}, the values that may occupy a participant position and denote an object. We assume, as for any database, that an object with a key can be fetched by it.

\begin{definition}[Operations]
\label{def:ops}
For a relationship $r$: \op{Store}: $r$ is an object with a persistent key. \op{Retrieve}: $r$ can be fetched by that key. \op{Attach}: properties can be stored on $r$ itself and read back. \op{Query}: $r$ can be selected by predicates and returned as a result. \op{Compose}: a reference to $r$ can occupy a participant position of another relationship.
\end{definition}

\begin{definition}[Identity and first-class relationships]
\label{def:fc}
A representation gives relationships \emph{identity} (\ID) if \op{Store}, \op{Retrieve} and \op{Query} hold for every relationship. It has \emph{first-class relationships} (\FC) if, taking its relationship construct to be all the objects that stand for relationships, \ID{} and \op{Attach} hold for every relationship, a single reference domain covers every relationship (so \op{Compose} holds for every relationship), and relationships of every arity $n \ge 2$ are available.
\end{definition}

\begin{definition}[Native and built constructs]
\label{def:built}
A host data model has a \emph{native} relationship construct: rows of a table with foreign keys to one table (SQL), edges between nodes (property graphs), reified triples (RDF~1.2), relations (TypeDB). A reference is \emph{built} when the host model does not provide it as a reference to relationships: a key into a supertable whose rows stand for relationships of several types, a \texttt{(table, id)} pair, or a node standing for a relationship. A construct is \emph{built} when its \op{Compose} depends on built references. A requirement is \emph{met by a built construct} when every representation of it in the host model uses one.
\end{definition}

\noindent SQL's native construct gives \ID{} when each relationship table declares a key, at any arity. It lacks \op{Compose}, because a foreign key refers to one table. A relationship supertable referenced by every type is a built construct with \op{Compose}; so is the polymorphic \texttt{(table, id)} pair above, without the database checking it.

\subsection{Four requirements}

Each requirement says what must be \emph{stored}. Section~\ref{sec:comparison} separately considers the query semantics a system gives the stored values.

\begin{definition}[\pillar{1}: n-ary relationships with attributes]
\label{def:p1}
For every $n \ge 2$, relationships among $n$ participants can be represented, with their attributes stored on the relationship, not on a participant or a pair of participants. \emph{Example:} \rel{access} joins three parties and carries Dana's permission level.
\end{definition}

\begin{definition}[\pillar{2}: temporal validity]
\label{def:p2}
Every relationship carries its valid-time intervals; updates preserve history; and a relationship that holds over several disjoint intervals is one relationship with several intervals. \emph{Example:} if Dana is rehired in June and her access restored, \rel{access} holds again, and ``how long has Dana held this access in total'' is a question about \rel{access}.
\end{definition}

\begin{definition}[\pillar{3}: confidence with an evidential basis]
\label{def:p3}
A relationship $r$ may carry a confidence $\kappa(r) \in [0,1]$, stored on $r$; and for every relationship $r$, an assessment of $r$ can be represented as a relationship whose participants include $r$, the assessment's source and its method, with the assessed value as an attribute (the \emph{evidential-basis clause}). \emph{Example:} \rel{offboarding} and \rel{signin} are the evidential basis for the confidence of \rel{access}.
\end{definition}

\begin{definition}[\pillar{4}: causal links between relationships]
\label{def:p4}
For any relationships $r_1, r_2$, the statement ``$r_1$ caused $r_2$'' can be represented as a relationship with $r_1$ and $r_2$ as participants, carrying attributes such as its kind and confidence. \emph{Example:} \rel{access} $\to$ \rel{approval}, of kind \emph{enabled}, confidence 0.9.
\end{definition}

\noindent The evidential-basis clause is a modelling choice with a reason: a confidence whose sources cannot be inspected cannot be combined, audited or recalibrated. In the example, a single number for \rel{access} would hide exactly the disagreement Q3 asks about.

\begin{remark}[Silent defaults]
\label{rem:defaults}
A representation that cannot store a period, a degree of belief or a cause in effect records a default---current, certain, uncaused. The practical consequence: \emph{not recorded} must be distinguishable from the default. A missing end may mean ``still holds'' or ``ended, but never recorded''; a missing confidence may mean ``certain'' or ``never checked''. In the running example, \rel{access} looked current precisely because its end was never recorded.
\end{remark}

\subsection{The Attributed Temporal Causal Hypergraph}

\begin{definition}[ATCH]
\label{def:atch}
An \emph{Attributed Temporal Causal Hypergraph} is a tuple $\mathcal{H} = (V, E, \mathit{Role}, \mathit{mem}, \mathcal{A}, \tau_v, \tau_t, \kappa)$ where $V$ is a finite set of entities; $E$ is a finite set of relationship identifiers; $\mathit{Role}$ is a finite set of role names; $\mathit{mem} \subseteq E \times (V \cup E) \times \mathit{Role}$ assigns each relationship its participants and their roles, which may include other relationships; $\mathcal{A}$ maps $V \cup E$ to finite sets of attribute--value pairs; $\tau_v$ and $\tau_t$ map each relationship to finite sets of valid-time and transaction-time intervals, whose end points may be finite or $+\infty$ (asserted to hold still), and, for valid time only, $\bot$ (not recorded); and $\kappa$ is a partial map $E \to [0,1]$, undefined where no confidence has been assessed. A \emph{causal link} is a relationship with exactly one participant in each of the roles \emph{cause} and \emph{effect}; an \emph{assessment} of $r$ is a relationship with $r$ in the role \emph{assessed}.
\end{definition}

\noindent ATCH has \FC{} natively: $E$ is one construct, its objects carry attributes, and $\mathit{mem}$ lets any of them be a participant. The example is an ATCH with nine relationships: four in Table~\ref{tab:example}, the two records, and three causal links. The confidence of \rel{access} is left undefined while its assessments conflict.

\begin{definition}[ATCH operators]
\label{def:ops-atch}
For an ATCH $\mathcal{H}$, a time $t$ and relationships $r, s$: $\mathrm{asof}_v(t) = \{e \in E : t \in I \text{ for some } I \in \tau_v(e)\}$, and $\mathrm{asof}_t$ likewise over $\tau_t$. An interval ending in $\bot$ contains $t$ only when $t$ is at or after its start and before every end that an assessment of $e$ records; whether it contains a later $t$ is unknown, and $\mathrm{asof}_v$ leaves it out. Further, $\mathrm{causes}^{*}(r)$ is the set of relationships from which $r$ is reachable along causal links; $\mathrm{assess}(r)$ is the set of assessments of $r$; and the confidence of a causal path is the product of its links' confidences.
\end{definition}

\section{The Threshold Theorem}
\label{sec:theorem}

\begin{lemma}[Attachment]
\label{lem:attach}
Let $r$ have $n \ge 3$ participants and an attribute $a$ that belongs to $r$. If $a$ is stored once, it is stored on an object that stands for $r$.
\end{lemma}

\begin{proof}
Suppose $a$ is stored on an object that does not stand for $r$. If the object stands for a participant, or for a binary relationship among two participants, then it does not determine $r$'s other participants, so $a$ is not recoverable as an attribute of $r$ unless it is also stored elsewhere. Storing $a$ once therefore means storing it on an object determined by all of $r$'s participants, that is, one standing for $r$.
\end{proof}

\begin{theorem}[Threshold]
\label{thm:threshold}
For any representation:
\begin{enumerate}
    \item[(a)] Each of \pillar{2}--\pillar{4} implies \ID; \pillar{1} implies \ID{} when relationships with equal participants and attributes must be kept distinct. \pillar{1}, \pillar{2} and \pillar{3} also imply \op{Attach}.
    \item[(b)] \pillar{3} and \pillar{4} each imply \op{Compose} for every relationship. Together with \pillar{1}, each implies \FC.
    \item[(c)] \FC{} suffices for \pillar{1}--\pillar{4}: one first-class construct, together with entities for participants, sources and methods and append-only updates of valid time, stores the data of all four.
\end{enumerate}
\end{theorem}

\begin{proof}
(a) \pillar{1}: by Lemma~\ref{lem:attach}, each n-ary relationship is an object carrying its attributes (\op{Attach}). Where relationships with equal participants and attributes must be kept distinct, the object is not determined by its contents, so it has a key of its own, which gives \op{Store} and \op{Retrieve}; selecting relationships by attribute gives \op{Query}. For binary relationships the same holds, since their attributes must be stored on the relationship by definition. \pillar{2}: a relationship valid over two disjoint intervals is one object carrying both, so it is an object (\op{Store}, \op{Retrieve}) whose key is not an interval, with its intervals stored on it (\op{Attach}); selecting relationships valid at a time is \op{Query}. \pillar{3}: a confidence or an assessment of $r$ refers to $r$, so $r$ is an object with a key (\op{Store}, \op{Retrieve}), and a confidence is stored on it (\op{Attach}); selecting by confidence or by assessment is \op{Query}. \pillar{4}: ``$r_1$ caused $r_2$'' refers to $r_1$ and $r_2$, so each is an object with a key; asking for the causes of $r_2$ returns $r_1$ (\op{Query}).

(b) \pillar{4}: for arbitrary $r_1, r_2$, the causal link is a relationship whose participant positions hold references to $r_1$ and $r_2$. Since $r_1$ and $r_2$ range over all relationships, the reference domain of the causal link's construct covers every relationship: \op{Compose} holds for every relationship. \pillar{3}: the same argument applies to the assessed position of an assessment, which by the evidential-basis clause can be filled by any relationship. With \pillar{1}, every arity $n \ge 2$ is available, and \op{Attach} and \op{Query} hold by (a). These are the conditions of Definition~\ref{def:fc}.

(c) Given \FC, store \pillar{1} data as relationships with attributes; the valid-time intervals of $r$ as a set-valued attribute of $r$, updated only by adding intervals or closing an open one, so that history is kept; each assessment as a relationship among $r$, its source and its method, with its value as an attribute, and $\kappa(r)$ on $r$; and each causal link as a relationship with $r_1$, $r_2$ in the roles \emph{cause} and \emph{effect}. Distinct data give distinct objects or attributes, and the data are recovered by selecting relationships by role and reading their attributes.
\end{proof}

\begin{corollary}[What must be built]
\label{cor:built}
If a host model's native construct lacks \FC, every representation in that model meeting \pillar{1} together with \pillar{3} or \pillar{4} has a first-class construct whose \op{Compose} depends on built references.
\end{corollary}

\begin{proof}
By Theorem~\ref{thm:threshold}(b) the representation has \FC. Its \op{Compose} reaches relationships of every type; the native construct's references do not, so some reference it uses is built.
\end{proof}

\noindent In SQL, the relationship supertable and the polymorphic \texttt{(table, id)} pair of Section~\ref{sec:example} are two different ways of meeting the corollary; one of them, or another built reference, is unavoidable.

\begin{remark}
\label{rem:scope}
Parts (a) and (b) follow almost directly from Definitions~\ref{def:ops}--\ref{def:p4}: \pillar{3} and \pillar{4} are stated so that they mention references to arbitrary relationships, and Lemma~\ref{lem:attach} recovers what \pillar{1} assumes. Their content lies in the choice of definitions, which Section~\ref{sec:example} motivates, and in the map they yield. That map locates the threshold precisely: identity is needed by every requirement, composition only by \pillar{3} and \pillar{4}, and first-class status by \pillar{1} with either of them; part (c) shows nothing beyond first-class status is needed to store the data. Definition~\ref{def:built} defines ``built'' by example, so Corollary~\ref{cor:built} is close to a restatement of it: a single \texttt{object} table referenced by a native foreign key gives SQL \op{Compose}, and calling that reference built is our classification, not a theorem. Formalizing representations as information-preserving schema mappings~\cite{hull1986relative}, so that the corollary becomes a separation result, is the main open problem.
\end{remark}

\begin{remark}[Where composition comes from]
\label{rem:compose}
\pillar{1} and \pillar{2} do not imply \op{Compose}: SQL tables with keys meet \pillar{1}, and can store several intervals per row, without any relationship referring to another of a different type. \pillar{3} implies \op{Compose} only through its evidential-basis clause; a bare confidence number on a relationship needs only \ID{} and \op{Attach}.
\end{remark}

\paragraph{Machine-checked part.} A Lean~4 development models a representation with one relationship construct whose participant slots may or may not hold relationships, and defines \op{Compose} as the existence of a slot value from which each relationship can be read back. It proves, for weakened forms of \pillar{3} (one assessment per relationship; source, method and value not modelled) and \pillar{4} (no roles or attributes), that each implies \op{Compose}, which is immediate from the definitions; that \op{Compose} together with a relationship builder yields representations of \pillar{3} and \pillar{4}; and that a construct of n-ary rows in which no slot can denote a row has a builder but no \op{Compose}, hence no \pillar{4}; that construct has no row-denoting slot by construction, so it does not model SQL, where a key can denote a row. Corollary~\ref{cor:built} is not formalized. In the Lean model, \FC{} is \op{Compose} together with a builder, which does not store attributes: it does not include \op{Store}, \op{Retrieve}, \op{Query}, \op{Attach}, arity or confidence values. Part~(a) and those aspects are argued on paper only; Lemma~\ref{lem:attach} is supported by a formalized counting fact (a binary edge among three distinct participants omits one of them). The development type-checks with Lean~4.14.0 and uses only the core library; it is available at \url{https://github.com/Network-Services-Group/equivalence-theorem-lean4}, where the tag \texttt{v1} keeps the formalization that accompanied version~1.

\section{Comparison}
\label{sec:comparison}

We distinguish three levels at which a host model supports a requirement. It is \emph{built} ($\circ$) when the requirement can be met only through a built construct (Definition~\ref{def:built}). It is \emph{representable} ($\triangle$) when the model's native construct stores the requirement's data directly, without query semantics for them. It is \emph{native} ($\checkmark$) when the model also gives the data built-in semantics: n-ary relationships as a primitive (\pillar{1}); time-travel queries over valid and transaction time (\pillar{2}); queryable assessments and a stated rule propagating confidence along causal chains (\pillar{3}); and causal traversal as a query primitive (\pillar{4}). These criteria for $\checkmark$ are the features ATCH was designed to have, so the $\checkmark$ column favours it by construction; ATCH is a model, and its implementation (Section~\ref{sec:queries}) is a built construct in SQL.

\begin{figure}[ht]
\centering
\begin{tikzpicture}[
  band/.style={rectangle, minimum height=1.05cm, text width=4.3cm, align=center, font=\small},
  sys/.style={font=\scriptsize, align=center}
]
\node[band, fill=red!8] (enc) at (0,0) {\textbf{Built}\\[-1pt]{\scriptsize first-class construct built by the schema}};
\node[band, fill=yellow!14] (rep) at (4.75,0) {\textbf{Representable}\\[-1pt]{\scriptsize stored by the native construct}};
\node[band, fill=green!12] (nat) at (9.5,0) {\textbf{Native}\\[-1pt]{\scriptsize stored, with query semantics}};
\draw[very thick, black] (2.375,-0.85) -- (2.375,0.85);
\node[font=\scriptsize, above] at (2.375,0.85) {native first-class construct (\FC)};
\node[sys, below=2pt of enc] {SQL, property graphs, RDF~1.2};
\node[sys, below=2pt of rep] {TypeDB};
\node[sys, below=2pt of nat] {ATCH};
\end{tikzpicture}
\caption{Support for confidence with an evidential basis (\pillar{3}), and for causal links between n-ary relationships (\pillar{1} with \pillar{4}). By Theorem~\ref{thm:threshold}(b), each needs a reference that covers every relationship; a model to the left of the line has no such native reference, so it meets these only through built ones, and with \pillar{1} through a built first-class construct (Corollary~\ref{cor:built}).}
\label{fig:threshold}
\end{figure}

\begin{table}[ht]
\centering
\small
\begin{tabular}{@{}lcccccc@{}}
\toprule
Host model (native construct) & \ID & \FC & \pillar{1} & \pillar{2} & \pillar{3} & \pillar{4} \\
\midrule
SQL (rows, foreign keys) & yes & no & $\checkmark$ & $\triangle$ & $\circ$ & $\circ$ \\
Property graphs (edges) & binary only & no & $\circ$ & $\triangle$ & $\circ$ & $\circ$ \\
RDF 1.2 (reified triples) & binary only & no & $\circ$ & $\triangle$ & $\circ$ & $\circ$ \\
TypeDB 3 (relations) & yes & yes & $\checkmark$ & $\triangle$ & $\triangle$ & $\triangle$ \\
ATCH (hyperedges) & yes & yes & $\checkmark$ & $\checkmark$ & $\checkmark$ & $\checkmark$ \\
\bottomrule
\end{tabular}
\caption{Support for the four requirements: $\circ$ built, $\triangle$ representable, $\checkmark$ native. \pillar{4} is rated for causal links between relationships of any arity; RDF~1.2 links binary statements natively. $\triangle$ on \pillar{2} for property graphs and RDF~1.2 holds for binary relationships only; n-ary ones are built. ATCH's ratings are for the model; its implementation is built in SQL. Consistently with Theorem~\ref{thm:threshold}(b) and Corollary~\ref{cor:built}, every model whose native construct lacks \op{Compose} for every relationship rates $\circ$ on \pillar{3} and \pillar{4}.}
\label{tab:comparison}
\end{table}

\paragraph{SQL.} A row is an n-ary relationship with attributes and a key, so SQL meets \pillar{1} natively. Several intervals per relationship can be stored in a child table keyed by the relationship, or as a multirange where supported, so \pillar{2} is representable. SQL:2011 adds application-time periods and system-versioned tables queried with \texttt{FOR SYSTEM\_TIME AS OF}~\cite{kulkarni2012temporal}; since an application-time row holds one period, a relationship valid over disjoint periods is several rows, which we still rate representable. Assessments and causal links refer to relationships of any type and need a built construct: the supertable or the polymorphic pair of Section~\ref{sec:example}. SQL/PGQ (ISO 9075-16:2023)~\cite{sqlpgq2023} adds property-graph queries over tables and keeps binary edges.

\paragraph{Property graphs.} Edges have keys and properties but are binary and cannot be endpoints of other edges~\cite{angles2018property}; GQL (ISO/IEC 39075:2024)~\cite{gql2024} standardizes this model. Valid time can be an edge property, so \pillar{2} is representable for binary relationships. The three-party \rel{access}, the records and the causal links become nodes standing for relationships: a built construct.

\paragraph{RDF 1.2.} RDF~1.2, a W3C Candidate Recommendation since April 2026~\cite{rdf12}, adds triple terms in object position, used with \texttt{rdf:reifies}: a \emph{reifier} stands for a statement and can be the subject or object of other triples. Reifiers give binary statements identity and composition, so a binary statement such as ``Dana is in the admin group'' can itself be a cause. \rel{access} and the records are n-ary and need the n-ary relation pattern, a built construct; so does any causal link with an n-ary end.

\paragraph{TypeDB.} Relations are n-ary, own attributes and can play roles in other relations~\cite{typedb,dorn2024typeql}, so TypeDB's native construct has \FC{} and stores the whole example directly. Valid time and confidence are ordinary attributes: TypeDB documents no as-of operator over valid time; its functions could compute confidence propagation, but none is built in.

\paragraph{ATCH.} With the operators of Definition~\ref{def:ops-atch}, ATCH answers Q1 by selecting from $\mathrm{assess}(\rel{access})$ on each record's value, Q2 with causal paths under the product rule, and Q3 with $\mathrm{assess}(\rel{access})$; $\mathrm{asof}_v$ answers the same question for any relationship whose valid time is settled. It keeps conflicting assessments rather than merging them into one confidence. These semantics are design choices; Theorem~\ref{thm:threshold} says only that a first-class construct can store what they need.

\section{The Investigators' Queries}
\label{sec:queries}

Evaluating conjunctive patterns is NP-complete in combined complexity, where data and query are both inputs~\cite{chandra1977optimal}. The investigators' questions fall in classes known to be polynomial. For a pattern over an ATCH, form its hypergraph with one hyperedge per pattern relationship, over the variables that relationship binds.

\begin{proposition}[Tractable classes]
\label{prop:tractable}
Let $|\mathcal{H}|$ be the size of an ATCH and $|q|$ the size of a query. The following are evaluable in time polynomial in $|\mathcal{H}|$ and $|\mathit{output}|$ for every fixed bound on the arity of pattern relationships, and also in $|q|$ when the hypergraph of membership atoms is $\alpha$-acyclic:
\begin{enumerate}
    \item patterns whose hypergraph is $\alpha$-acyclic~\cite{fagin1983degrees}, with selections on attributes, valid-time windows and confidence thresholds;
    \item causal reachability: the relationships reachable from $r$ along causal links, with no depth bound;
    \item the most confident causal path from a relationship $s$ to a relationship $t$, where chain confidence is the product of the link confidences along the path.
\end{enumerate}
\end{proposition}

\begin{proof}
(1) A pattern relationship with $k$ roles is a base relation of at most $|\mathcal{H}|^{k}$ tuples materialized from $\mathit{mem}$, polynomial for bounded $k$; acyclic patterns over these relations admit Yannakakis evaluation~\cite{yannakakis1981algorithms}, polynomial in input and output; selections apply to base relations before joins. The hypergraph of membership atoms has one hyperedge per atom $\mathit{mem}(e, x, \rho)$ of the pattern instead of one per pattern relationship; when it is $\alpha$-acyclic, the same evaluation runs over $\mathit{mem}$ itself, which has at most $|\mathcal{H}|$ tuples, so nothing is materialized and the bound is polynomial in $|q|$ as well. (2) Breadth-first search over causal links, linear in $|\mathcal{H}|$. (3) With link confidences in $(0,1]$, maximizing a product is minimizing $\sum -\log\kappa$, a shortest-path problem with non-negative weights, solved by Dijkstra's algorithm.
\end{proof}

\noindent Q1 is a pattern of two relationships, an assessment and \rel{access}, joined through \rel{access} in the role \emph{assessed}, with a selection on the assessment's value: an acyclic pattern, class (1). Q2 is class (3) with $s = \rel{access}$ and $t = \rel{payment}$. Q3 selects the relationships with \rel{access} in the role \emph{assessed}.

\begin{remark}[Enumerating paths]
Listing every causal path between two relationships can produce exponentially many paths in the size of $\mathcal{H}$; the cost is in the output, not in deciding reachability.
\end{remark}

\begin{remark}[Confidence cut-off]
\label{rem:depth}
Simple causal paths have at most $|E|$ links. Tighter: if every causal link reachable from $r$ has confidence at most $\kappa_{\max} < 1$, a path of $d$ links has chain confidence at most $\kappa_{\max}^{d}$, so a traversal that keeps paths with chain confidence at least $\theta \in (0,1)$ keeps paths of at most $\lfloor \log\theta / \log\kappa_{\max} \rfloor$ links. If some reachable link has confidence~1, only the bound $|E|$ applies.
\end{remark}

\paragraph{Implementation.} ATCH has a prototype implementation as a PostgreSQL extension, by exactly the built construct of Corollary~\ref{cor:built}: a supertable \texttt{pdb\_hyperedges} holds every relationship, with its attributes, confidence, and valid time as a range type under a GiST index; \texttt{pdb\_hyperedge\_members} holds entity participants and \texttt{pdb\_hyperedge\_member\_edges} relationship participants, each with a role. Causal links, between distinct relationships and kept acyclic for every writer and under concurrent writes, are stored with their confidence in \texttt{pdb\_causal\_links}, and any link can be promoted to a relationship of its own (\texttt{pdb\_causal\_link\_hyperedge}) whose participants are its cause and effect. Transaction time is queried with \texttt{pdb\_as\_of}; every change to a relationship's valid time is kept (\texttt{pdb\_valid\_time\_history}), and so is every change to its label, attributes and confidence, so that \texttt{pdb\_as\_of} and \texttt{pdb\_belief\_as\_of} return what was believed at the instant asked rather than what is believed now. Against Definition~\ref{def:atch}, the implementation falls short in three ways. A relationship has one valid-time interval; the history records superseded values of it (earlier beliefs), not further periods in which it held, so \pillar{2}'s several-interval clause is not yet met. A confidence can be stored as unassessed, and a path through an unassessed relationship then has unknown confidence; but a confidence omitted when a relationship is written still defaults to 1, the silent default Remark~\ref{rem:defaults} warns against. And causal-path confidence also multiplies each effect's own confidence and a per-link context modifier, with paths enumerated up to a bound rather than the most confident one computed as in Proposition~\ref{prop:tractable}(3). Measured performance and an evaluation on operational data are future work.

\section{Related Work}
\label{sec:related}

\paragraph{Data model expressiveness.} Comparisons of what data models can represent are made precise through information capacity~\cite{hull1986relative}: one schema dominates another when an information-preserving mapping exists. Our levels in Section~\ref{sec:comparison} are informal; restating them in terms of information capacity is future work. Property-graph schemas~\cite{angles2023pgschema} and pattern calculi~\cite{deutsch2022graph,francis2023gpc} formalize the binary model we compare against. Meta-property graphs~\cite{sadoughi2024metapg} make labels and properties first class and support reification, without addressing time, confidence or causation.

\paragraph{Relationships as objects.} The ER model has n-ary relationships with attributes~\cite{chen1976entity}, and aggregation lets relationships take part in relationships~\cite{smith1977database}; ORM objectification and UML association classes do the same in conceptual modelling~\cite{halpin2008information}. RM/T gives every association a surrogate drawn from one domain~\cite{codd1979extending}, which is the relationship supertable of Section~\ref{sec:example}. Object identity~\cite{khoshafian1986object}, its use as a query primitive~\cite{abiteboul1989object} and value invention in ILOG~\cite{hull1990ilog} supply the identifiers that \ID{} and \op{Compose} need. Topic Maps~\cite{iso13250} have n-ary, role-typed associations that can be reified, and HyperGraphDB~\cite{iordanov2010hypergraphdb} stores links that point to links. None of these is new here. What they leave open is which storage requirements force such a construct, and that question is what Theorem~\ref{thm:threshold} answers.

\paragraph{Statements about statements.} For RDF, the n-ary relation patterns~\cite{noy2006nary}, RDF*~\cite{hartig2017foundations} and the singleton property~\cite{nguyen2014singleton} are compared for Wikidata by Hern\'andez, Hogan and Kr\"otzsch~\cite{hernandez2015reifying}. Wikidata's references and ranks give each statement an evidential basis~\cite{vrandecic2014wikidata}, close to \pillar{3}. Annotated RDF attaches time, fuzzy confidence and provenance to triples~\cite{udrea2010annotated,zimmermann2012annotated}. PROV records qualified influences between activities~\cite{moreau2013prov}, and Trio combines uncertainty with lineage~\cite{benjelloun2006uldb}.

\paragraph{N-ary and hyper-relational knowledge.} Hypergraphs generalize graphs to multi-party relationships~\cite{berge1989hypergraphs}. Wikidata attaches qualifiers to statements~\cite{vrandecic2014wikidata}, and hyper-relational knowledge-graph embeddings learn over qualified statements~\cite{galkin2020stare}. HyperGraphRAG~\cite{hypergraphrag2025} reports that hypergraph-structured retrieval outperforms binary graphs for retrieval-augmented generation. Temporal hyper-relational formalisms combine n-ary structure, time and second-order edges~\cite{ahrabian2026htkgh}, without confidence or causal links between relationships.

\paragraph{Time.} Bitemporal databases~\cite{jensen1999temporal,snodgrass2000developing} and Allen's interval algebra~\cite{allen1983maintaining} underlie our use of valid and transaction time; temporal property-graph models~\cite{debrouvier2021tgql} attach time to nodes and edges. Temporal elements~\cite{gadia1988homogeneous} let one tuple hold over several disjoint intervals, as \pillar{2} asks, and SQL:2011 standardizes application and system time~\cite{kulkarni2012temporal}. ATCH attaches time to n-ary relationships that can take part in other relationships.

\paragraph{Confidence and provenance.} Probabilistic databases model uncertainty over possible worlds~\cite{suciu2011probabilistic}. Provenance semirings~\cite{green2007provenance} give an algebra for how annotations of inputs combine into annotations of query results. Why- and where-provenance~\cite{buneman2001why} characterize which inputs a result came from. ATCH's evidential basis is a stored form of provenance, and its confidence rules are one choice of combination that a semiring treatment would generalize.

\paragraph{Causation and reasoning.} Pearl's causal models~\cite{pearl2009causality} infer effects among variables; ATCH stores asserted causal links between relationships and does not infer them. Situation and event calculi~\cite{mccarthy1969some,kowalski1986logic,reiter2001knowledge} reason about change; ATCH can serve as the store over which such reasoning runs.

\section{Conclusion}

Records of who could access what, what changed and what caused what are records of relationships. Storing them directly needs a threshold: every requirement needs relationship identity; confidence with an evidential basis and causal links between relationships also need composition; and with n-ary relationships, either needs first-class relationships, which in turn suffice for all four. Under our definitions the threshold belongs to the relationship construct, not the host model, so models whose native construct falls short meet the requirements by building one, as our SQL implementation does. Turning that classification into a separation result over information-preserving mappings is the main open problem.

The result says what must be stored, not how it is obtained. The open question is practical: whether real operational data can be extracted into this form, and whether confidences can be measured rather than assigned. We are evaluating extraction from IT operations systems, and the calibration of confidences produced by decision models, using the operator's systems of record as the answer key and technician audits of samples.

\enlargethispage{\baselineskip}
\section*{Acknowledgments}
The author used Claude (Anthropic) as an AI research assistant for formalization, literature review and manuscript preparation. All theoretical contributions, research direction, proof validation and editorial decisions are the author's, who takes full responsibility for the content.

\end{document}